\documentclass{article}
\usepackage[tbtags]{amsmath}
\usepackage{amsfonts,amssymb,mathrsfs,amscd}
\usepackage{amsthm}
\usepackage{graphicx}

\usepackage{tikz}  
\usetikzlibrary{calc}
\usetikzlibrary{positioning,arrows,shapes,backgrounds}

\newcommand{\kk}{\mathbb{K}}
\newcommand{\FF}{\mathbb{F}}
\newcommand{\ZZ}{\mathbb{Z}}
\newcommand{\QQ}{\mathbb{Q}}

\newcommand{\CC}{\mathbb{C}}
\newcommand{\HH}{\mathcal{H}} 
\newcommand{\mat}[1]{\left[\begin{array}{cc}#1\end{array}\right]}
\newcommand{\zdim}[3]{[#1/#2\,| #3]}

\newcommand{\ncap}{} 

\newcommand{\Gr}{\mathrm{Gr}}
\newcommand{\join}{\vee} 
\newcommand{\meet}{\wedge} 
\newcommand{\opg}{\boldsymbol{\cdot}}

\newcommand{\img}{\mathop{\mathrm{im}}} 
\newcommand{\hclass}[1]{\overline{#1}}

\theoremstyle{plain}
\newtheorem{all}{}[section]
\newtheorem{proposition}[all]{Proposition} 
\newtheorem{theorem}[all]{Theorem} 
\newtheorem{lemma}[all]{Lemma} 
\newtheorem{corollary}[all]{Corollary} 

\theoremstyle{definition}
\newtheorem{definition}[all]{Definition}
\newtheorem{example}[all]{Example} 
\newtheorem{remark}[all]{Remark} 
\newtheorem{exercise}[all]{Exercise} 
\newtheorem{notation}[all]{Notation} 
\newtheorem{construction}[all]{Construction}

\date{}
\begin{document}

\title{\bf\large On the Cleaning Lemma of Quantum Coding Theory}

\author{\normalsize Gleb Kalachev\footnote{Lomonosov Moscow State University, 
kalachev.gleb2@gmail.com},
\hspace{1em} Sergey Sadov\footnote{serge.sadov@gmail.com}}

\maketitle

\noindent
{\small
{\bf Abstract.}
The term ``Cleaning Lemma'' refers to a family of similar propositions that have been used in Quantum Coding Theory
to estimate the minimum distance of a code in terms of its length and dimension. We show that the mathematical core is a simple fact of linear algebra
of inner product spaces; moreover, it admits a further reduction to a combinatorial, lattice-theoretical level.
Several concrete variants of the Cleaning Lemma and some additional propositions are derived as corollaries of the proposed approach.
}

\medskip\noindent
{\small

\medskip\noindent
%
} 

\section{Introduction}

\subsection{Goal: to make simple look simple}

The purpose of this paper is to reveal a simple mathematical nature of the so-called Cleaning Lemma (CL). This is a device employed in the theory of quantum error correcting codes (QECC) to obtain 
fundamental bounds on the  minimum distance for certain classes of codes implied by the prescribed physical limitations. 

We interpret CL as a consequence of a proposition of elementary linear algebra, where 
the distracting mathematical features usually present in the literature on QECC, such as the Pauli matrices, the characteristic 2, etc.,
play no role. The relevant context is just inner product spaces.
(A skeleton of the idea can be seen in the chain: Proposition~\ref{prop:factor-isom} ---
Theorem~\ref{thm:InnProd} --- Proposition~\ref{prob:CL-stabilizer}.)

The most abstract prototype of CL is stated at a combinatorial level, involving lattice-theoretical terminology.

CL is found in varying contexts and varying formulations. 
It is, properly speaking, a qualitative result, saying that a certain set of qubits can be ``cleaned'' of logical operators. However, a primary technical result that is usually proved first is a quantitative generalization: it is a reciprocity relation of the form $\ell+\ell'=\mathrm{const}$ between certain, ``dual'' to each other, numerical parameters of the code. The phenomenon of ``cleaning'' properly occurs when one of the numbers $\ell$ or $\ell'$ equals zero. 
A proof is typically short, so one can guess that the result is rather simple; however, because of a rich context it is not easy to single out the essential elements.

One can draw parallels between CL and Fredholm's Alternative (FA), which, in the qualitative form says that a system of linear equations with a square matrix has a solution
for any right-hand side if and only if the homogeneous system has only the trivial solution. The quantitative generalization is the equality of the row and  column ranks or, equivalently, the equality to zero of the index of any finite-dimensional linear operator. In fact, there is more than mere analogy: both CL and FA (even over a non-commutative field) are consequences of a common lattice-theoretical proposition. 
Furthermore, FA can be interpreted at an yet higher level involving a canonical construction --- specifically, the standard short exact sequence. We will mention a certain canonical construction related to CL, too. 

Our intention is to facilitate the dissemination of CL in expository and teaching practice aimed at students and general mathematical audience. 
Reader's expertise in Quantum Computing
or even in classical coding theory is not a prerequisite here. 
There is no physics in this paper and mathematics is very basic. 
However, CL should be appreciated in the context of its applications, which inevitably 
assume much more elaborate background 
and subject-specific terminology. A very sketchy review of such applications is given in the next subsection.
An accessible introduction to the topic, suitable for a  mathematician non-expert in quantum codes, can be found in \cite{Haah2016}.
As a general reference on quantum computing we mention the textbook \cite{KitaevShenVyalyi} and the article \cite{Kribs2005}, while a variety of  topics on quantum error correction is covered in \cite{LaGuardia2020}.

The structure of the paper is outlined in the last subsection of this Introduction. One can legitimately ask why is the paper claiming to offer a ``simple'' approach so long.
The answer is, while hoping that our analysis may open new ways of looking at CL and related things, we do not know what can bear fruit. Hence we present a variety of formulations, seek interconnections, and go into some digressions.

\subsection{Cleaning Lemma as a tool for estimating parameters of quantum codes}

CL has been primarily used to derive upper and lower bounds on the minimum distance for certain classes of quantum codes. Let us recall some notions related to quantum error correction.

Interaction of a quantum system, such as a quantum computer, with environment necessarily causes a gradual decoherence of the quantum state. At the logical level this corresponds to the occurrence of errors.
Quantum error correction is a way to ensure a robust logical functioning of the system.
Much as in the classical coding theory, a redundancy is introduced: a system of $k$ logical qubits is encoded by a system of $n>k$ physical qubits. 
Mathematically, a quantum $[[n,k]]$-code is an embedding $(\CC^2)^{\otimes k}\subset (\CC^2)^{\otimes n}$; the numbers $k$ and $n$ are, respectively, the code dimension and the code length. Each factor in the tensor product
$(\CC^2)^{\otimes n}$ represents a single physical qubit.  The embedded space is referred to as the {\em code space}.

Further concepts, concerning error detection and correction, will be described in the context of the widely used class of {\em stabilizer codes}. 

In a stabilizer code, the code space 
is comprised by invariant vectors of a set of commuting self-adjoint unitary operators (the {\em stabilizers}), each one being the tensor product of Pauli matrices acting on their respective qubits.
One can verify whether the quantum state belongs to the code by a procedure of syndrome measurement, where stabilizers are used as measurement operators. 
Errors are also described by unitary operators. Occurrence of errors may cause a failure of some of the vectors in the code space to be invariant; this can be detected by the measurement.

Logical errors are the unitaries for which the code space is invariant; for this reason they are called {\em logical operators}. A logical operator acting trivially (as a scalar operator) on the code space is safe: the logical state remains intact; such errors are called {\em degenerate}. The danger lies with errors that are neither detectable (i.e.\ leave the code space invariant) nor acting trivially on the code space.

An important parameter of a quantum code is the {\em minimum distance}. It is the least number of errors
in individual qubits that result in a logical error. The greater the minimum distance, the more fault-tolerant the code is.

Let $M$ be a set of physical qubits. A logical operator is {\em supported on $M$}, if it acts on the complement of $M$ trivially.
The set $M$ is said to be a {\em correctable region}\ if any logical operator supported on $M$ is degenerate. 

The measurement procedure by itself is a source of errors. Therefore it is important to keep the number of operations involved in syndrome measurement small.
For stabilizer codes this corresponds to the requirement that the number of qubits changed by any stabilizer be bounded by some constant.

In the existing architectures of quantum computers, two-qubit operations are available only when the qubits are spatially adjacent.
This leads to the additional requirement:
the qubits acted upon by a stabilizer must be situated close to each other. Such codes are called (\emph{geometrically}) {\em local}.

The constraints on the main parameters of the code --- dimension and distance --- implied by locality were studied in the papers
\cite{Bravyi&Terhal:2009,Bravyi:UpperDistDim}.

In \cite{Bravyi&Terhal:2009},
where CL was apparently first introduced,
it was used to obtain an upper bound on the code distance for local quantum codes.
The proof of the upper bound goes by way of contradiction. Assuming that the code distance equals $d$, one shows that sets of qubits inside a cube with boundary of cardinality $\le\alpha d$ are correctable; thus for a sufficiently large $d$ the whole set of qubits turns out to be correctable, contrary to the fact that the dimension of the code is nonzero.

In the paper \cite{Bravyi:UpperDistDim}
an upper bound on the code dimension in terms of the minimum distance and length of a 2D local quantum code is obtained.
Here an additional idea is employed, which allows one to show that if the whole set of qubits is partitioned into a set $C$ and a pair of correctable sets $A$, $B$,
then the code dimension cannot exceed $|C|$, cf.~Proposition~\ref{prop:codedim-abc}.
In the case of stabilizer codes the construction of appropriate correctable sets $A$ and $B$ can be ensured by CL. 
In fact, \cite{Bravyi:UpperDistDim} 
goes beyond stabilizer codes and employs a generalization of CL, the ``Disentangling lemma'' (DL), which involves additional structure and does not fit in the concise formalism developed here. Unlike with CL, no ``quantitative'' form of DL has been given.

In the paper \cite{baspin2021:treewidth-dist}
the above mentioned technique for proving the upper bounds from \cite{Bravyi&Terhal:2009,Bravyi:UpperDistDim}
is generalized to allow arbitrary stabilizer codes, not necessarily possessing the property of locality. As a substitute for locality,
the treewidth of the connectivity graph of the quantum code is used. In the connectivity graph one can define the boundary of a qubit set and use CL to prove that the minimum distance does not exceed the treewidth by order of magnitude. 

CL can also been used to obtain a lower estimate for the code distance $d$. To this end, it suffices to show that for any qubit subset $M$ of cardinality 
$|M|< d$ the dimension of the space of logical operators supported on the complement of $M$ is equal to the dimension of the space of all logical operators; then CL asserts that the dimension of the space of logical operators supported on $M$ equals zero, that is, $M$ is a correctable region.

The outlined approach has been used 
\cite{Ioshida2013:FractalQCodes} to obtain lower estimates of the form $d\ge \mathrm{const}\cdot L$ for fractal quantum codes where qubits are positioned at the nodes of the cube lattice with sidelength $L$. 

Many constructions of quantum CSS codes are borrowed from algebraic topology or can be naturally interpreted in geometric terms. 
Example include the toric code \cite{Kitaev:toric:1997}, the homological product codes \cite{Bravyi:HMP:2014}, and recently introduced quantum code families: fiber bundle codes \cite{Hastings:2020:fiber}, lifted product codes \cite{LiftedProduct} and balanced product codes \cite{BalancedProduct}. 
Conversely, a number of results of the Quantum Coding Theory, suitably interpreted, lead to new results in geometric topology. For instance, in \cite{Freedman:qcodes2manifold} first examples of manifolds exhibiting the power law $\ZZ_2$-systolic freedom
have been constructed using recent bounds \cite{Hastings:2020:fiber,LiftedProduct} on the code distance.
This is a deep subject and we do not go into it; yet, to illustrate the use of the homological dictionary, in Sec.~\ref{ssec:QCSS-chain} we formulate one of the versions of CL in terms of a chain complex.

\subsection{The paper structure}

There are four types of material in the paper:

\begin{itemize}
\item 
the Cleaning Lemma as such, in different guises (Section~\ref{sec:CL});
\item
related linear-algebraic constructions
(Section~\ref{sec:iso} and \S~\ref{ssec:lalg});
\item
an abstract, lattice-theoretical core
(\S~\ref{ssec:comb});
\item
extras and miscellania (Section~\ref{sec:misc}).
\end{itemize}

In order to avoid mixing mathematical and quantum-computing terminology,
we postpone the formulations of CL until after the mathematical groundwork is completed. 

Section~\ref{sec:iso} is purely linear-algebraic. We deal with a vector space,
its dual, and three subspaces. An equivalent 
treatment in terms of a simplectic space and  its subspaces is then given.
The picture will have a feel of {\em d\'{e}j\`{a}~vu}\ for many readers. 
And one can take a shortcut from Proposition~\ref{prop:factor-isom} to Theorem~\ref{thm:InnProd}. 

Alternatively, Section~\ref{sec:genrk}
offers a slow but more general 
approach. The purpose of \S~\ref{ssec:comb}, where
we descend to the combinatorial level, is to 
demonstrate that even the linear structure is somewhat redundant for CL-related identities. 
The main abstract result is Theorem~\ref{thm:Agraded-2lat}. The subsequent results of Sec.~\ref{sec:genrk} are, in essence, its specializations.

The  part of theory applicable to QECC lies at the linear-algebraic level (\S~\ref{ssec:lalg}) culminating with Theorem~\ref{thm:InnProd}, from which there is just one step to CL: translation into the language that fits in with literature on QECC.

Notice that Section~\ref{sec:genrk} deals only with numerical (rank) identities; there are no canonical isomorphisms. 

In Section~\ref{sec:CL} CL at last appears as a technical statement. 
Dictionaries rather than logical inferences play the dominant role here. We discuss versions of CL for stabilizer codes (\S~\ref{ssec:QSC}), CSS codes (\S~\ref{ssec:QCSS}), and subsystem codes (\S~\ref{ssec:SubC}). In addition, in Subsection~\ref{ssec:QCSS-chain} we recast the CSS version of CL in terms of a 
chain complex.

The last Section~\ref{sec:misc} is a sampling of other applications, arguably selected quite randomly, of the formalism developed in Sec.~\ref{sec:genrk}.    

First, in \S~\ref{ssec:misc-FA}
we treat the Fredholm Alternative from this point of view. Then, in \S~\ref{ssec:misc-QC}
we turn again to quantum codes and prove two 
propositions, one new and one known --- but in a coordinate-free form.
Finally, in \S~\ref{ssec:misc-groups}, we present a toy structure involving a 
product of Abelian groups where analogs of correctable regions and ``cleaning'' can be observed in a mathematically different situation. 

Besides the mentioned applications, little digressions and extras are interspersed throughout as Exercises. The proofs should present no difficulty.

\subsection*{Preliminaries: notation, the Grassmanian, modularity}
If $V$ is a vector space (over some field), we denote by $\Gr(V)$ the Grassmanian of $V$ (the set of all subspaces). The relation of set inclusion makes $\Gr(V)$ a partially ordered set.  Moreover, it is a modular 
lattice. (Basic definitions concerning lattices can be found, e.g., in \cite[Ch.~3]{Stanley1986}.) 
Subspaces (and elements of general lattices 
in \S~\ref{ssec:comb})
will be denoted by Greek letters.
The lattice operations
`join' ($\join$) and `meet' ($\meet$) correspond in the lattice $\Gr(V)$ to the sum
($+$) and intersection ($\cap$) of subspaces.
We will usually omit the symbols $\meet$ and $\cap$ and write $\alpha\beta$ instead of $\alpha\land\beta$
or $\alpha\cap\beta$.

Recall that the operation $\cap$ in $\Gr(V)$ does not  generally 
distribute 
over
$+$. Instead, the weaker {\em modular law}\ 
holds: if $\alpha\subset\gamma$, then
$(\alpha+\beta)\gamma=\alpha+\beta\gamma$.

\section{A canonical construction}
\label{sec:iso}

\subsection{The picture in a space and its dual}
\label{ssec:iso-dual}

Let $V$ be a vector space (over some field $\kk$, which we assume fixed throughout) and $V^*=\mathrm{Hom}(V,\kk)$ be the dual vector space. 
An action of $f\in V^*$ on $x\in V$ will be written as the coupling: 
$f(x)=\langle f,x\rangle$.
If $\xi\in \Gr(V)$, we denote by $\xi^\bot = \{f\in V^*\mid \langle f,x\rangle=0\mbox{ for all }x\in\xi\}$ the annihilator of $\xi$.
Similarly, for a subspace $\eta\subset V^*$
its annihilator in $V$ is denoted by $\eta^\bot$.

If $\xi\subset V$, $\eta\subset\xi^\bot\subset V^*$, then there is the canonical isomorphism $\phi:(\xi^\bot/\eta)\to (\eta^\bot/\xi)^*$, $\phi(f+\eta):(x+\xi) \mapsto 
\langle f,x\rangle
$ for $f\in \xi^\bot$, $x\in \eta^\bot$. 
Below we identify $(\xi^\bot/\eta)$ and $(\eta^\bot/\xi)^*$ assuming that we use the isomorphism $\phi$.

\begin{proposition}
\label{prop:factor-isom}
Let $V$ be a vector space, $\xi\subset V$ and $\eta\subset V^*$ be subspaces of $V$ and $V^*$, respectively. Suppose that  $\xi\subset\eta^\bot$ (equivalently, $\eta\subset\xi^\bot$). For any subspace $\alpha\subset V$ the vector space $\alpha_1=(\eta^\bot\ncap\alpha+\xi)/\xi$ is a subspace of $\eta^\bot/\xi$ and the vector space $\alpha_2=(\xi^\bot\ncap\alpha^\bot+\eta)/\eta$ is a subspace of $\xi^\bot/\eta$. The subspace $\alpha_2$ is the annihilator of $\alpha_1$: 
$$
 \alpha_2=\alpha_1^\bot.
$$
\end{proposition}
\begin{proof}
    For any subspace $\zeta\subset V$ such that $\xi\subset\zeta\subset\eta^\bot$ we have 
    \begin{equation}
    \label{eqn:factor-compl}
        (\zeta/\xi)^\bot=\{f+\eta\mid f\in V^*: \langle f,x\rangle=0\;\mbox{ for all }\,x\in \zeta\}=\zeta^\bot/\eta,
    \end{equation}
    where the right-hand side is a subspace of $\xi^\bot/\eta$.
    
    Taking $\zeta=\eta^\bot\ncap\alpha+\xi$, we get $\zeta/\xi=\alpha_1$. Further,
    $\zeta^\bot=(\eta+\alpha^\bot)\ncap\xi^\bot=\eta+\alpha^\bot\ncap\xi^\bot$. Here we use the assumption $\eta\subset\xi^\bot$ and the modular law in $\Gr(V^*)$. Hence
    $\zeta^\bot/\eta=\alpha_2$. Reference to  \eqref{eqn:factor-compl} completes the proof.
   \end{proof}

From here the reader who looks for a quick way can go straight to
Remark~\ref{rem:canonical-numerical}
in  \S~\ref{ssec:lalg} 
that links  Proposition~\ref{prop:factor-isom} to
Theorem~\ref{thm:InnProd}.

Proposition~\ref{prop:factor-isom}
will be cited one more time in 
\S~\ref{ssec:QCSS-chain}.

In the next subsection we give an equivalent treatment of the above construction in terms of a symplectic space, where the space $V$ and its dual $V^*$ become Lagrangian subspaces. It will not be referred to later on, but in some circumstances the ``symplectic'' view may turn preferential.

\subsection{The symplectic picture}
\label{ssec:iso-symp}
Let us introduce the skew-symmetric inner product on the space $W=V\oplus V^*$ by
$$
 (x_1\oplus y_1, x_2\oplus y_2)=\langle x_1,y_2\rangle-
 \langle x_2,y_1\rangle,
$$
where $x_1,x_2\in V$ and $y_1,y_2\in V^*$.
Thus $W$ becomes a symplectic space. 
We can and will consider $V$ and $V^*$ as its complementary Lagrangian subspaces.
It is said that $W= V\oplus V^*$ is a Lagrangian decomposition \cite[p.~24]{MaslovBook}. Conversely, a pair of complementary Lagrangian subspaces in a symplectic space can be interpreted as a pair of mutually dual vector spaces.

The operation of taking an orthogonal subspace in $W$ will be denoted by the symbol ${}^\#$ to distinguish it from the operation ${}^\bot$ of taking an annihilator in the dual space. 
For instance, $V^\#=V$ and $( V^*)^\#=V^*$.

Now we reproduce the construction of \S~\ref{ssec:iso-dual} within the symplectic space~$W$.

We begin with isomorphism $\phi$.
Given the subspaces
$\xi\subset V$
and $\eta\subset V^*$ such that $\xi\subset\eta^\bot$, 
consider the isotropic subspace 
$I_{\xi,\eta}=\xi\oplus\eta$ of $W$. 
We have
$$
 I_{\xi,\eta}^\#=
 \xi^\#\eta^\#
 =
 \eta^\bot\oplus\xi^\bot.
$$

The subspace $I_{\xi,\eta}$ is the radical (kernel) of the space $I_{\xi,\eta}^\#$
(with inner product induced from $W$).
Factorization yields the space $I_{\xi,\eta}^\#/I_{\xi,\eta}$
with nondegenerate symplectic product
\cite[Ch.III,\S~3]{Artin1958}.
The symplectic decomposition
$I_{\xi,\eta}^\#/I_{\xi,\eta}\cong
(\eta^\bot/\xi)\oplus(\xi^\bot/\eta)$, corresponds to the isomorphism $\phi$.

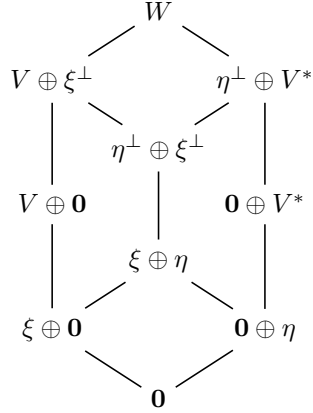
\begin{figure}
    \centering
    \begin{tikzpicture}[-,node distance=2cm, semithick, auto]
    \newcommand{\dy}{0.5cm}
      \tikzstyle{state}=[fill=white,text=black,yshift=\dy]
      \node[state] (W)                              {$W$};
      \node[state] (xib)  [below left of = W]       {$V\oplus\xi^\bot$};
      \node[state] (etab) [below right of = W]      {$\eta^\bot\oplus V^*$};
      \node[state] (etabxxib) [below right of = xib] {$\eta^\bot\oplus\xi^\bot$};
      \node[state] (xixeta) [below of = etabxxib]   {$\xi\oplus \eta$};
      \node[state] (xi)   [below left of = xixeta]  {$\xi\oplus \mathbf{0}$};
      \node[state] (eta)  [below right of = xixeta] {$\mathbf{0}\oplus\eta$};
      \node[state,yshift=-\dy] (v)    at ($(xi)!0.5!(xib)$)     {$V\oplus \mathbf{0}$};
      \node[state,yshift=-\dy] (v')   at ($(etab)!0.5!(eta)$)   {$\mathbf{0}\oplus V^*$};
      \node[state] (zero) [below right of = xi]     {$\mathbf{0}$};
      \path[draw] (W)--(xib)--(v)--(xi)--(zero);
      \path[draw] (W)--(etab)--(v')--(eta)--(zero);
      \path[draw] (xib)--(etabxxib)--(xixeta) -- (xi);
      \path[draw] (etab)--(etabxxib);
      \path[draw] (xixeta) -- (eta);
    \end{tikzpicture}
    \caption{A lattice obtained from two complementary Lagrangian subspaces in a symplectic space, their two subspaces and respective annihilators}
    \label{fig:symplectic-lattice}
\end{figure}

Figure~\ref{fig:symplectic-lattice} gives the first hint at 
the underlying combinatorics of subspaces. It shows the lattice of subspaces of the space $W$ that have appeared so far.   

\smallskip
To carry out the construction of 
Proposition~\ref{prop:factor-isom} in the symplectic context, we will need two lemmas.
The proofs are left as an exercise.

\begin{lemma}
\label{lem:lagr-decomp}    
To any subspace
$\zeta\subset V$ there corresponds the Lagrangian subspace 
$
 \Lambda_\zeta=\zeta\oplus\zeta^\bot
$
in $W=V\oplus V^*$.
Conversely, if for some $\zeta\subset V$ and $\zeta'\subset V^*$ the subspace
$\zeta\oplus\zeta'\subset W$ is Lagrangian, then $\zeta^\bot=\zeta'$.
\end{lemma}

\begin{lemma}
\label{lagr-factor}
{\rm (Cf.~\cite[Lemma~1.4.38(a)]{MaslovBook}.)}
If $\Lambda$ is a Lagrangian subspace and $\alpha\subset\Lambda\subset\alpha^\#$, 
then $\Lambda/\alpha$ is a Lagrangian subspace in $\alpha^\#/\alpha$.
\end{lemma}

\begin{construction}
\label{cons:sympl}
Let us take an arbitrary subspace $\alpha\subset V$.
Consider the Lagrangian subspace 
$
 \Lambda_\alpha=
 \alpha\oplus\alpha^\bot
$
corresponding to $\alpha$ 
as in the first part of Lemma~\ref{lem:lagr-decomp}.
Consider the subspace $\zeta\subset V$
and its orthogonal complement 
$\zeta^\bot\subset V^*$ as defined below:
$$
\zeta=\xi+\eta^\bot\ncap\alpha,
\quad
\zeta^\bot=\eta+\xi^\bot\ncap\alpha^\bot.
$$
The subspace $\Lambda_\zeta=\zeta\oplus\zeta^\bot$ is again Lagrangian (in $W$).
By Lemma~\ref{lagr-factor}, the image 
$\Lambda_{\zeta}/I_{\xi,\eta}$
of $\Lambda_\zeta$ in the symplectic factorspace
$I_{\xi,\eta}^\#/I_{\xi,\eta}$
is 
a Lagrangian subspace, too. It admits the decomposition
$\Lambda_{\zeta}/I_{\xi,\eta}\cong (\zeta/\xi)\oplus(\zeta^\bot/\eta)$.
Therefore,
by the second part of Lemma~\ref{lem:lagr-decomp},
the orthogonality relation
$(\zeta/\xi)^\bot=\zeta^\bot/\eta$ holds --- the same as in Proposition~\ref{prop:factor-isom}.
\hfill{$\Box$}
\end{construction}

In the final part of this section we want to present a symplectic analog of Proposition~\ref{prop:factor-isom} that has a particularly elegant form, see \eqref{quotient-factor} below. The idea is to avoid an explicit reference to the Lagrangian decomposition $W=V\oplus V^*$ and instead to regard  as the primary data the isotropic subspace $I_{\xi,\eta}$ (which will become $\sigma$ below) and the Lagrangian subspace $\Lambda_\alpha$ (which will become $\beta$ and may not even be Lagrangian).

Let us fix an arbitrary isotropic subspace $\sigma$ of $W$. Define the map
$$
 h_\sigma:\;\Gr(W)\to\Gr(W),
\quad
 h_\sigma(\beta)=\sigma+\sigma^\#\beta.
$$
Its main properties are summarized in the next lemma.

\begin{lemma}
\label{lem:lagr-iso}
(a) For any $\beta$, we have $\sigma\subset h_\sigma(\beta)\subset\sigma^\#$. 

(b) The map $h_\sigma$ commutes with involution ${}^\#$, that is, $(h_\sigma(\beta))^\#=h_\sigma(\beta^\#)$.
In particular, if $\beta$ is isotropic (resp., Lagrangian), then so is
$h_\sigma(\beta)$.
\end{lemma}

\begin{proof}
The property (a) is obvious. For (b) we have
 $$
  (\sigma+\sigma^\#\ncap\beta)^\#=\sigma^\#\ncap(\sigma^\#\ncap\beta)^\#=\sigma^\#\ncap(\sigma+\beta^\#)
  =\sigma+\sigma^\#\ncap\beta^\#.
 $$
 Here, as in the proof of Proposition~\ref{prop:factor-isom}, the last step relies on
 the modular law and the inclusion $\sigma\subset \sigma^\#$.
\end{proof}

Due to the property (a), we have the well-defined map 
$$
 q_\sigma:\;\Gr(W)\to\Gr(\sigma^\#/\sigma),
\quad
 q_\sigma(\beta)=h_\sigma(\beta)/\sigma.
$$

The space $\sigma^\#/\sigma$ is the maximal subspace
of $W/\sigma$ in which the inner product induced from $W$
by the natural formula $(x+\sigma, y+\sigma)=( x,y)$
is correctly defined. The induced inner product is nondegenerate, so the operation of taking the orthogonal subspace is an involution in $\Gr(\sigma^\#/\sigma)$. We denote it
by the same symbol ${}^\#$ as in $W$.

\begin{proposition}
\label{prop:lagr-isom}
For any subspace $\beta\in\Gr(W)$ we have
\begin{equation}
\label{quotient-factor}
 (q_\sigma(\beta))^\#=q_\sigma(\beta^\#),
\end{equation}
where ${}^\#$ on the left pertains to $\sigma^\#/\sigma$
and on the right --- to $W$.
\end{proposition}

\begin{proof}
It follows at once from Lemma~\ref{lem:lagr-iso}. \end{proof}

Construction~\ref{cons:sympl} relates to this Proposition via the correspondences
$I_{\xi,\eta}\mapsto\sigma$, $\Lambda_\alpha\mapsto\beta$, 
 $\Lambda_\zeta\mapsto h_\sigma(\beta)$, and, finally,
$\Lambda_\zeta/I_{\xi,\eta}\mapsto q_\sigma(\beta)$.

\begin{remark}
The map $h_\sigma$ preserves the order:
$\beta_1\subset\beta_2\;\Rightarrow\;h_\sigma(\beta_1)\subset h_\sigma(\beta_2)$, but it is not a lattice homomorphism: in general,
$h_\sigma(\beta\gamma)\neq h_\sigma(\beta)h_\sigma(\gamma)$. The same applies to the map $q_\sigma$.
\end{remark}

\begin{exercise}
\label{ex:properties-h}
This exercise exhibits some further properties of
the map $h_\sigma$ in addition to those listed in Lemma~\ref{lem:lagr-iso}.

(a) The map $h_\sigma$ is idempotent.

(b) $h_\sigma(\beta)\, \beta=\sigma^\#\beta$.

(c) The image under $h_\sigma$ of the set of all Lagrangian subspaces in $W$ is the set of all 
Lagrangian extensions of the isotropic subspace $\sigma$.
\end{exercise}

\begin{exercise}
It is instructive to compare the property (b) in Exercise~\ref{ex:properties-h} with Lemma~1.4.39 in \cite{MaslovBook}, which asserts that, given a Lagrangian subspace $\lambda$, there exists a Lagrangian subspace $ \lambda'$ such that $ \lambda'\ncap \lambda=\sigma\ncap \lambda$.
Derive this result from Ex.~\ref{ex:properties-h}(b). Hint: consider a Lagrangian subspace 
$\lambda^*$
complementary to $\lambda$ (that is, $\lambda\ncap\lambda^*=0$) and the isotropic
subspace
$\sigma'=(\sigma^\#+\lambda)\lambda^*$. 
Put
$\lambda'=h_{\sigma'}(\lambda)$. 
\end{exercise}

\section{Quasi-complementations and the general rank identity}
\label{sec:genrk}

By the general rank identity we understand the identity of Theorem~\ref{thm:Agraded-2lat} and any of its specializations
in Theorems~\ref{thm:Agraded-2}, \ref{thm:3subspaces}, and \ref{thm:InnProd}.

\subsection{Combinatorial formulation}
\label{ssec:comb}

Let $L$ be a lattice with $\bot$ and $\top$ its minimum and maximum elements, respectively.
The partial order is denoted by $\leq$ as usual; the join and meet operations
are $\join$ and $\meet$,
respectively; the symbol $\land$ is usually omitted.

\begin{definition} 
\label{def:grlat}
(A group-graded lattice.)
Let $L$ be a lattice and $G$ an Abelian group with operation $\opg$. 
A {\em $G$-grading} 
on $L$ is a map $L\to G$, $x\mapsto|x|$, satisfying the condition
\begin{equation}
\label{gmodular}
    |x\join y|\opg |xy|=|x|\opg|y| \text{ for all $x,y\in L$}.
\end{equation}    
A lattice $L$ with a given $G$-grading will be called a {\em $G$-modular}\ lattice.
\end{definition}

If there is no confusion as to which $G$-grading on $L$ is being used in the current context, then
we will refer to $|x|$ as to the rank of the element $x\in L$. 
If $L$ is a modular lattice (in the usual sense), then its (usual) rank function is a $\ZZ$-grading.
In this situation, we will use the additive notation.
(However, there exist other $\ZZ$-gradings.)

\begin{exercise} 
All $G$-gradings on $L$ constitute an Abelian group $\mathrm{Grad}(L,G)$
--- a subgroup of the group $G^L$ of all $G$-valued functions on $L$. 

The $G$-gradings with property $|\bot|=1_G$ constitute a subgroup $\mathrm{Grad}_0(L,G)$ and $\mathrm{Grad}(L,G)/\mathrm{Grad}_0(L,G)\cong G$. 
\end{exercise} 

\begin{notation} 
\label{not:gmodquot}
Let $L$ be a $G$-modular lattice. For any $x,y,z\in L$ denote
$$
 [y/x]=|y|\opg|x|^{-1}
$$
and
$$
\zdim{y}{x}{z}=[yz/xz].
$$

Given two $G$-graded lattices $L_1$ and $L_2$, we will 
write $[y/x]_i$ and $\zdim{y}{x}{z}_i$ whenever $x,y,z\in L_i$  $(i=1,2)$. 
\end{notation}

\begin{definition} 
\label{def:antiisolat}
A bijection $\dag:L_1\to L_2$ between two $G$-graded lattices is an {\em anti-isomorphism}\ if
it reverses the order ($x\leq_1 y\Rightarrow y^\dag\leq_2 x^\dag$) and 
\begin{equation}
\label{commonproduct}
|x|_1\opg|x^\dag|_2=|y|_1\opg|y^\dag|_2.
\end{equation}
for any $x,y\in L_1$.
\end{definition}

The inverse of the anti-isomorphism is an anti-isomorphism mapping $L_2$ onto $L_1$.
We will denote it by the same symbol ${}^\dag$ (thus we extend the map ${}^\dag$ to a bijection
of the set $L_1\sqcup L_2$ onto itself).

\begin{exercise} (Cf.~\cite[Lemma 83]{Rotman1998}.)
\label{ex:antiisolat}
If ${}^\dag:L_1\to L_2$  is an order-reversing bijection, then for any $x,y\in L_1$
\begin{equation}
\label{jdagm}
(x\join y)^\dag = x^\dag\meet y^\dag.
\end{equation}
\end{exercise}

\begin{theorem}
\label{thm:Agraded-2lat}
Let $L_1$, $L_2$ be two $G$-graded lattices and ${}^\dag:L_1\to L_2$ be an anti-isomorphism.
If
$\xi\in L_1$, $\eta\in L_2$, 
then for any $\alpha\in L_1$ the identity
$$
\zdim{\eta^\dag}{\xi}{\alpha}_1 \opg \zdim{\xi^\dag}{\eta}{\alpha^\dag}_2 = 
[\eta^\dag/\xi]_1 = [\xi^\dag/\eta]_2
$$
holds.
\end{theorem}

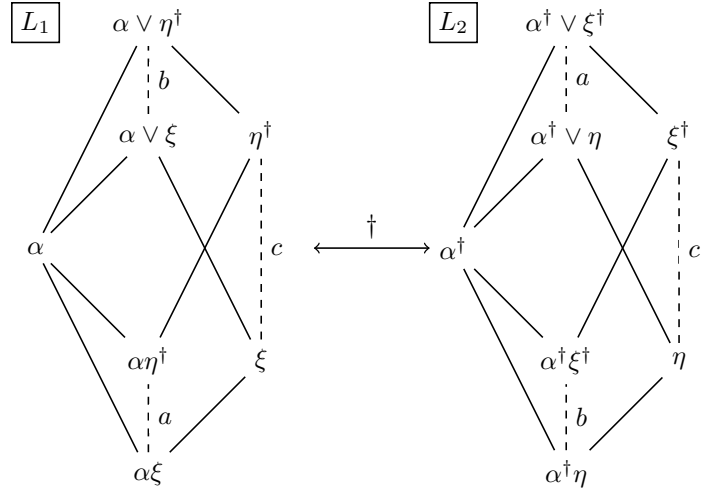
\begin{figure}[hbt]
    \centering
    \begin{tikzpicture}[-,auto,node distance=1.5cm,
                        semithick]
      \tikzstyle{state}=[fill=white,text=black]
      \node[state]        (a)                    {$\alpha$};
      \node[state]        (c) [right of = a] {};
      \node[state]         (avxi) [above of = c] {$\alpha\vee \xi$};
      \node[state]         (avn) [above of = avxi] {$\alpha\vee\eta^\dag$};
      \node[state]         (n) [right of = avxi] {$\eta^\dag$};
      \node[state,below of = c]         (an)       {$\alpha\eta^\dag$};
      \node[state]         (axi) [below of = an]       {$\alpha\xi$};
      \node[state]         (xi) [right of = an]       {$\xi$};
    
      \path[draw,dashed] (axi)--node[right,fill=white] {$a$}(an);
      \path[draw] (an)--(a)--(avxi); 
      \path[draw,dashed] (avxi)--node[right,fill=white] {$b$}(avn);
      \path[draw] (axi)--(a)--(avn);
      \path[draw] (axi)--(xi)--(avxi);
      \path[draw] (an)--(n)--(avn);
      \path[draw,dashed] (xi)--node[right,fill=white] {$c$}(n);
      \node[state,right=5cm of a]   (a')                    {$\alpha^\dag$};
      \node[state]        (c') [right of = a'] {};
      \node[state]         (avxi') [above of = c'] {$\alpha^\dag\vee \eta$};
      \node[state]         (avn') [above of = avxi'] {$\alpha^\dag\vee\xi^\dag$};
      \node[state]         (n') [right of = avxi'] {$\xi^\dag$};
      \node[state,below of = c']         (an')       {$\alpha^\dag\xi^\dag$};
      \node[state]         (axi') [below of = an']       {$\alpha^\dag\eta$};
      \node[state]         (xi') [right of = an']       {$\eta$};
    
      \path[draw,dashed] (axi')--node[right,fill=white] {$b$}(an');
      \path[draw] (an')--(a')--(avxi'); 
      \path[draw,dashed] (avxi')--node[right,fill=white] {$a$}(avn');
      \path[draw] (axi')--(a')--(avn');
      \path[draw] (axi')--(xi')--(avxi');
      \path[draw] (an')--(n')--(avn');
      \path[draw,dashed] (xi')--node[right,fill=white] {$c$}(n');
      \coordinate[right = 0.7cm of a -| xi] (ra);
      \path[<->,draw] (ra) -- node{$\dag$}(a');
      \node[draw, left of =avn]{$L_1$};
      \node[draw, left of =avn']{$L_2$};
    \end{tikzpicture}
    \caption{To the proof of Theorem~\ref{thm:Agraded-2lat}. Solid lines connect comparable elements in the lattices; dashed lines correspond to the quotients of ranks
    used in the proof.}
    \label{fig:2lattices}
\end{figure}

\begin{proof}
Put $\zdim{\eta^\dag}{\xi}{\alpha}_1=a$,
$[\alpha\join\eta^\dag/\alpha\join\xi]_1=b$,
$[\eta^\dag/\xi]_1=c$. 
We write $\frac{x}{y}$ for $x\opg y^{-1}$.
By property \eqref{gmodular} in $L_1$, we have 
$$
a\opg b=\frac{|\alpha\eta^\dag|_1}{|\alpha\xi|_1} \opg \frac{|\alpha\join\eta^\dag|_1}{|\alpha\join\xi|_1}
=\frac{|\alpha|_1\opg|\eta^\dag|_1}{|\alpha|_1\opg|\xi|_1}
=\frac{|\eta^\dag|_1}{|\xi|_1}=c.
$$
By \eqref{commonproduct}, $[\eta^\dag/\xi]=[\xi^\dag/\eta]$.
Finally, by \eqref{commonproduct} and \eqref{jdagm}, 
$$
b=\frac{|\alpha\join\eta^\dag|_1}{|\alpha\join\xi|_1}
=\frac{|(\alpha\join\xi)^\dag|_2}{|(\alpha\join\eta^\dag)^\dag|_2}
=\frac{|\alpha^\dag\xi^\dag|_2}{|\alpha^\dag\eta|_2}=\zdim{\xi^\dag}{\eta}{\alpha^\dag}_2.
$$
The relations constituting the proof are visualized in Figure~\ref{fig:2lattices}.
\end{proof}

\begin{definition} 
\label{def:gmodcomp}
An involutive anti-isomorphism of a $G$-modular lattice $L$ onto itself is called {\em quasi-complementation}.
\end{definition}

\begin{exercise}
\label{ex:modinvrk}
Suppose that $L$ is a modular lattice and $|\cdot|$ its (usual) rank function. Then any order-reversing involution
$\dag$ on $L$ is a quasi-complement\-ation. Specifically, $|x|+|x^\dag|=|\top|$ for any $x\in L$. 
(Cf.\ \eqref{commonproduct}.)
\end{exercise}

\begin{definition} 
\label{def:gmodort}
Two elements $x$ and $y$ of a $G$-modular lattice with quasi-complementation ${}^\dag$ are {\em orthogonal}
(notation: $x\bot y$) if $x\leq y^\dag$, equivalently, $y\leq x^\dag$.
\end{definition}

Note that $x\bot x^\dag$ for any $x$.

Except for \S~\ref{ssec:misc-FA} we will need Theorem~\ref{thm:Agraded-2lat} in the case  $L_1=L_2$. We give an explicit
formulation for this special, the most important case. 

\begin{theorem}
\label{thm:Agraded-2}
Let $L$ be a $G$-modular lattice with quasi-complementation $\dag$.
For any $\xi$, $\eta$ and $\alpha$ in $L$ there holds the identity
$$
\zdim{\eta^\dag}{\xi}{\alpha}\opg\zdim{\xi^\dag}{\eta}{\alpha^\dag} 
=[\eta^\dag/\xi]
= [\xi^\dag/\eta].
$$
In particular, 
$$
\zdim{\xi^\dag}{\xi}{\alpha} \opg \zdim{\xi^\dag}{\xi}{\alpha^\dag} = [\xi^\dag/\xi].
$$
\end{theorem}

\subsection{Linear-algebraic formulation}
\label{ssec:lalg}

In this section we 
apply the preceding theory to the modular lattice of subspaces of a vector space.

Let $V$ be a vector space of dimension $n$ over a field $\kk$.
Recall that the Grassmanian $\Gr(V)$
is a modular lattice with rank function $|\lambda|=\dim \lambda$, which takes values $0,1,\dots,n$.

A {\em quasi-complementation}\footnote{In this definition we follow \cite{Kubiak1999}, p.~436.}\ is an order-reversing involution on $\Gr(V)$. That is, ${}^\dag: \Gr(V)\to\Gr(V)$ is a quasi-complementation if 
$(\lambda^\dag)^\dag=\lambda$ and
$
 (\lambda\subsetneq\mu) \,\Rightarrow\,(\lambda^\dag\supsetneq\mu^\dag)
$ for any subspaces $\lambda$ and $\mu$ of $V$.

By Exercise~\ref{ex:modinvrk} the notion of quasi-complementation as defined here agrees with
notion of quasi-complementation in the modular lattice $\Gr(V)$ in the sense of Definition~\ref{def:gmodcomp}.

\begin{notation}
\label{not:linasd}
Given two subspaces $\alpha$ and $\beta$ of $V$, denote
$$
\rho(\alpha,\beta)=|\alpha/(\alpha\beta)|=|(\alpha+\beta)/\beta|. 
$$
\end{notation}

\begin{remark} The symmetrized function 
$d(\alpha,\beta)=\rho(\alpha,\beta)+\rho(\beta,\alpha)=|\alpha+\beta|-|\alpha\beta|$
is the {\em subspace metric}\ widely used in coding theory.
\end{remark}

\begin{theorem}
\label{thm:3subspaces}
Let $\dag$ be a quasi-complementation on $\Gr(V)$ and $\alpha,\xi,\eta\in\Gr(V)$. 

(a) Suppose that $\xi\subset\xi^\dag$.
Then 
$$
\rho(\alpha\xi^\dag,\xi)+\rho(\alpha^\dag\xi^\dag,\xi)=|\xi^\dag|-|\xi|=n-2|\xi|.
$$

(b) 
Suppose that subspaces $\xi$, $\eta$ satisfy $\xi\subset\eta^\dag$, $\eta\subset\xi^\dag$. Then 
$$
\rho(\alpha\eta^\dag,\xi)
+\rho(\alpha^\dag\xi^\dag,\eta)
=
n-|\xi|-|\eta|.
$$
\end{theorem}

\begin{proof} (a) is a particular case $\xi=\eta$ of (b).
 Since $\eta^\dag\xi=\xi$, we have, comparing Notation~\ref{not:gmodquot} and Notation~\ref{not:linasd},
$$
\rho(\alpha\eta^\dag,\xi)=|\alpha\eta^\dag/\alpha\eta^\dag\xi|
=|\alpha\eta^\dag/\alpha\xi|=\zdim{\eta^\dag}{\xi}{\alpha},
$$
and similarly $\rho(\alpha^\dag \xi^\dag,\eta)=\zdim{\xi^\dag}{\eta}{\alpha^\dag}$.
In view of the equalities $|\xi^\dag|+|\xi|=|\eta^\dag|+|\eta|=n$, the identity in (b) is a specialization of Theorem~\ref{thm:Agraded-2}.
\end{proof}

\begin{exercise}
\label{ex:2subspaces}
Verify the identities 
$
\rho(\alpha,\beta)=\rho(\beta^\dag,\alpha^\dag)
$
and $|\alpha\beta|-|\alpha^\dag\beta^\dag|=|\alpha|-|\beta^\dag|
$
for any $\alpha,\beta\in\Gr(V)$.
Use the latter to prove Theorem~\ref{thm:3subspaces} directly, without reference to Theorem~\ref{thm:Agraded-2}.
\end{exercise}

Quasi-complementations in $\Gr(V)$ most naturally arise from bilinear or 
(for $\kk=\CC$)
sesquilinear forms. Admitting some redundancy in exposition, we explicitly state the main identity in this case.

\smallskip
By an {\em orthogonal space}\ we understand a finite-dimensional vector space $V$ over a field $\kk$ equipped with
``inner product'' $\langle \cdot,\cdot\rangle$: $V\times V\to \kk$ of either
of the following three types:

\begin{itemize}
\item  the form $\langle \cdot,\cdot\rangle$ is bilinear and symmetric;
\item the form $\langle \cdot,\cdot\rangle$ is bilinear and antisymmetric;
\item $\kk=\CC$ and the form $\langle \cdot,\cdot\rangle$ is Hermitian.
\end{itemize}

For 
$\lambda\in V$, its orthogonal complement is $\lambda^\bot=\{x\in V\mid \forall y\in\lambda\;\langle x,y\rangle=0\}$.

Clearly, the orthogonal complement of the zero subspace $\mathbf{0}$ is $V$.
The orthogonal space is {\em nondengenerate}\ if $V^\bot=\mathbf{0}$.
In an nondegenerate orthogonal space, the operation ${}^\bot$ is involutive and thus a quasi-complementation. 

Orthogonality of two subspaces, $\lambda\bot\mu$, means that $\lambda\subset\mu^\bot$; equivalently,
$\mu\subset\lambda^\bot$. This is consistent with lattice-theoretical notion of orthogonality (Definition~\ref{def:gmodort}). 

\begin{theorem}
\label{thm:InnProd}
(a) Let $(V,\langle \cdot,\cdot\rangle)$ be a nondegenerate orthogonal space of dimension $n$. If $\xi,\eta,\alpha\in\Gr(V)$
and $\xi\bot\eta$, then
\begin{equation}
\label{orthospace-sum-dimensions}    
 |(\eta^\bot\alpha+\xi)/\xi|+|(\xi^\bot\alpha^\bot+\eta)/\eta|
 =n-|\eta|-|\xi|.
\end{equation}
(b) In particular, if $\xi$ is an isotropic subspace ($\xi\subset\xi^\bot$), then
$$
|(\xi^\bot\alpha+\xi)/\xi|+|(\xi^\bot\alpha^\bot+\xi)/\xi|
 =|\xi^\bot|-|\xi|=n-2|\xi|.
$$
\end{theorem}

Recall that here $|U|=\dim U$.

\begin{remark}
\label{rem:canonical-numerical}
In 
the logic of this Section, Theorem~\ref{thm:InnProd} is a consequence of Theorem~\ref{thm:3subspaces}.
On the other hand, it can be also seen as a corollary of Proposition~\ref{prop:factor-isom}. Indeed, in the $n$-dimensional nondegenerate orthogonal space $(V,\langle\cdot,\cdot\rangle)$ the inner product defines an isomorphism between $V$ and $V^*$ and the orthogonal 
complement is the same as its annihilator in $V^*$.
From Proposition~\ref{prop:factor-isom} we know that the left-hand side in
\eqref{orthospace-sum-dimensions} is equal to $|\xi^\bot/\eta|$, which coincides with right-hand side.
\end{remark}

\begin{exercise}
\label{ex:isotropic1}
Let $V$ be a symplectic space of dimension $2m$. For any Lagrangian subspace $\alpha$
and any isotropic subspace $\xi$ the identity
$|\alpha\ncap\xi^\bot|-|\alpha\ncap\xi|=m-|\xi|=\frac{1}{2}(|\xi^\bot|-|\xi|)$ holds.
\end{exercise}

\section{Variants of the Cleaning Lemma}
\label{sec:CL}

\subsection{The case of Quantum Stabilizer Codes}
\label{ssec:QSC}

A description of Quantum Stabilizer Codes (QSC) involves, at the primary level, certain algebras of operators in a (finite-dimensional) Hilbert space. 
We will work at the secondary level --- a symplectic space model, in which operators and their commutation relations are substituted by vectors and their symplectic products, while all essential concepts, such as, for instance, the code space and 
logical operators, have their counterparts. 

\smallskip
Let us very briefly, just in order to make connection to the Quantum Codes literature, outline the constituents of the primary level:

\begin{itemize}
\item The $2^n$-dimensional ($n$-qubit) complex Hilbert space $\HH=H^{\otimes n}=\mathop{\otimes}\limits_{j=1}^n H^{(j)}$, where $H\cong\CC^2$. The number $n$ is the {\em length}\ of the code.
\item The algebra $\mathcal{L}(\HH)$ of linear endomorphisms of $\HH$.
\item The unitary group $U(\HH)$ and its subgroup $P(\HH)$ generated by tensor products of Pauli matrices and the scalar $iI$.
\item The {\em stabilizer group}\ $S$ --- a commutative subgroup of $P(\HH)$.
\item The {\em code space}\ --- the common eigenspace  with eigenvalue $1$ of operators from $S$.
\item The {\em algebra of logical (or error) operators}\ $L$ consists of operators for which the code space is an invariant subspace.
\item The {\em algebra of trivial
(degenerate)
logical operators}\ $L_0=\langle S\rangle\subset L$ consists of operators that 
act as scalars in the code space.
\item The {\em equivalence of logical operators}: $A,B\in L$, $A\sim B$ means that $B=AC$ for some $C\in S$.
\item A {\em 
region}\ $M\subset[n]$, where $[n]=\{1,2,\dots,n\}$ and its complement $M^c=[n]\setminus M$.
\item The subalgebra $\alpha=\otimes_{j\in M}\mathcal{L}(H^{(j)})$ in $\mathcal{L}(\HH)$ of {\em operators supported on $M$}\ and the complementary subalgebra
$\alpha^\dag$ of operators supported on $M^c$, so that $\alpha$ and $\alpha^\dag$ commute and 
$\alpha\otimes\alpha^\dag=\mathcal{L}(\HH)$.
\end{itemize}

The Cleaning Lemma for QSC asserts the alternative:
{\em
\begin{itemize}
    \item[---] 
    either there exists 
 a nontrivial logical operator supported on $M$, i.e.
 $L\cap \alpha\not\subset L_0$;
    \item[---] or any logical operator, up to equivalence,
     can be realized as an operator supported on $M^c$, i.e. for any $A\in L$ there exists $B\in  \alpha^\dag$ such that $A\sim B$.
\end{itemize}
} 

``In other words, a correctable region can be {\em cleaned}\ of any logical operators.'' \cite[p.~312]{Haah2016}.

\smallskip
The secondary level provides a simplification due to the fact that the Pauli group factored by scalar operators (one disregards the phase factors of wavefunctions) is commutative. The factorgroup is called the {\em reduced Pauli group}. Its elements,
being products of Pauli $X$-matrices in some of the spaces $H^{(j)}$ and, independently, Pauli $Z$-matrices in
some of the spaces $H^{(j)}$, can be encoded by binary words of length $2n$; the product in the reduced group 
corresponds to the componentwise sum of the words treated as vectors in $\FF_2^n$.
Morever, the corresponding Pauli operators (before reduction) commute or anti-commute depending on the value of a suitable symplectic product of the respective vectors in $\FF_2^n$. 

The table below presents a basic dictionary from the primary level to the secondary level. The right column shows the dimensions of the (sub)spaces involved in the secondary level.

\bigskip\noindent
\begin{tabular}{p{0.5\textwidth}|l|c}
\hline
Reduced Pauli group $P/\langle e^{2\pi i/4}\rangle$ & Symplectic space $P$ & $2n$ \\
Stabilizer group $S$\newline 
(trivial logical operators) &  Isotropic subspace $S\subset P$ & $n-k$\\
Its centralizer in $P$ & $S^\bot$ & $n+k$ \\
Logical operators up to equivalence & $S^\bot/S$ & $2k$ \\
Operators supported on 
region $M$ & Subspace $\alpha\subset P$ & \\ 
Operators supported on region $M^c$ & $\alpha^\bot$ & \\ 
Logical operators supported on $M$ & $(\alpha\cap S^\bot+S)/S$ & $\ell_M$ \\
Logical operators supported on $M^c$ & $(\alpha^\bot\cap S^\bot+S)/S$ & $\ell_{M^c}$ \\
\hline
\end{tabular}

\bigskip
At the secondary level, the Cleaning Lemma for QSC is the last part of the following proposition, cf.~\cite[Proposition 2.19 and Theorem~2.20]{Haah2016}.

\begin{proposition}
\label{prob:CL-stabilizer}
For a symplectic space $P$ and its subspaces as shown in the 2nd and 3rd columns of the above table,
the identity
$$
 \ell_M+\ell_{M^c}=2k
$$
holds.
In particular, if $\ell_M=0$, then $\ell_{M^c}=2k$, hence for any $v\in S^\bot$ there exists
$v'\in S^\bot\cap\alpha^\bot$ such that $v-v'\in S$.
\end{proposition}

\begin{proof} 
Apply Theorem~\ref{thm:InnProd}(b) with $V=P$ and $\xi=S$.
\end{proof}

\subsection{The case of Quantum CSS Codes}
\label{ssec:QCSS}

The Calderbank-Shor-Steane (CSS) codes are a particular case of stabilizer codes.
Let us restrict considerations in this case to the second (symplectic space) level right away and illustrate Proposition~\ref{prob:CL-stabilizer} using matrix representations.

Let $V$ be the vector space $\FF_2^n$ with inner product
$$
 (u,v)=\sum_{i=1}^n u_i v_i.
$$
The symplectic product in the space $P=V\oplus V$ is given by $(u\oplus u',v\oplus v')=(u,v')-(u',v)$.

Let $H_x$, $H_z$ be two matrices (over $\FF_2$) of sizes $m_1\times n$ and $m_2\times n$ respectively, such that
\begin{equation}
\label{Hxz0}
H_x H_z^T=0.
\end{equation}
Equivalently,
$$
 \mat{H_x & 0 \\ 0 & H_z}\,\mat{0 & I_n \\ I_n & 0}\,\mat{H_x^T & 0 \\ 0 & H_z^T}=0_{(m_1+m_2)\times (m_1+m_2)}.
$$

Let $\xi$ and $\eta$ be the row spaces of the matrices $H_x$ and $H_z$, respectively.
The elements of these subspaces are called the {\em degenerate codewords}.

The subspace $\xi\oplus\eta\subset P$ is the isotropic space $S$ as defined in Sec.~\ref{ssec:QSC}.

The subspace 
$\xi^\bot=\ker H_x$ is called the space of {\em $Z$-codewords}; similarly, $\eta^\bot=\ker H_z$ is
the space of {\em $X$-codewords}. Here ${}^\bot$ is treated as an involution in $\Gr(V)$ (rather than
in $\Gr(P)$).

The condition \eqref{Hxz0} is equivalent to $\xi\subset\eta^\bot$, $\eta\subset\xi^\bot$, or simply
$\xi\bot\eta$.

\begin{proposition}
\label{prop:CL-CSS}
\newcommand{\rank}{\mathop{\mathrm{rk}}}
Let $k$ be the dimension of a CSS code defined by the matrices $H_x$ and $H_z$, so that
$$
 k= |\ker H_x/\img H_z^T|=|\ker H_z/\img H_x^T|=n-\rank H_x-\rank H_z.
$$
Let $\alpha\subset V$ be any subspace. Consider the dimensions of spaces of equivalence classes of $X$- and $Z$- logical operators that have representatives in $\alpha$ $(\ell_x$ and $\ell_z)$ or $\alpha^\bot$ $(\ell'_x$ and $\ell'_z)$:
\begin{align*}
 \ell_{x} &=\left|(\ker H_z\cap \alpha+\img H_x^T)/\img H_x^T\right|,      &\ell_{z} &=\left|(\ker H_x\cap \alpha+\img H_z^T)/\img H_z^T\right|,\\
 \ell'_{x}&=\left|(\ker H_z\cap \alpha^\bot+\img H_x^T)/\img H_x^T\right|, &\ell'_{z}&=\left|(\ker H_x\cap \alpha^\bot+\img H_z^T)/\img H_z^T\right|.
\end{align*} 
Then
$$
\ell_{x}+\ell'_{z}=\ell_{z}+\ell'_{x}=k.
$$
\end{proposition}

\begin{proof}
Taking into account the identifications preceding the proposition, we apply
Theorem~\ref{thm:InnProd}(a) to the space $V$ 
(rather than $P$).
\end{proof}

In the context of CL we have a partition $[n]=M\sqcup M^c$ of the index set of the natural basis in $V$.
The subspace $\alpha\in V$ consists of vectors with zero $M^c$-components, while the subspace $\alpha^\bot$
consists of vectors with zero $M$-components.

We leave to the reader to see that, on the other hand, Proposition~\ref{prop:CL-CSS} 
in a cruder form, 
$\ell_{x}+\ell'_{x}+\ell_{z}+\ell'_{z}=2k$,
follows from Proposition~\ref{prob:CL-stabilizer}.

\medskip
The following artificially simple example may turn helpful for some readers.

\begin{example}  
Consider $V=\FF_2^n$, $n=2k$, $H_x=[I_k, I_k]$ ($I_k$ is the identity $k\times k$ matrix), and $H_z=0$. 
Let $M=\{1,2,\dots,k\}$, $M^c=\{k+1,\dots,2k\}$.

 Let us find all relevant subspaces explicitly. First, obviously,
  \begin{align*}
  \xi&=\{[x,x]\mid x\in \FF_2^k\},
   &
  \xi^\bot&=\{[x, -x] \mid x\in \FF_2^k\},
  \\
  \eta&=\{0\},
  &
  \eta^\bot&=V,
  \\
  \alpha&=\{[x,0]\mid x\in \FF_2^k\},
  &
  \alpha^\bot&=\{[0,x]\mid x\in \FF_2^k\}.
 \end{align*}
 Therefore 
  \begin{align*}
  (\eta^\bot\cap \alpha)+\xi&=\alpha+\xi=\{[x+y,y]\mid x,y\in\FF_2^k\}=V,\\
  (\eta^\bot\cap \alpha^\bot)+\xi&=\alpha^\bot+\xi=\{[x,x+y]\mid x,y\in\FF_2^k\}=V,\\
  (\xi^\bot\cap \alpha)+\eta&=\{0\}+\eta=\{0\},\\
  (\xi^\bot\cap \alpha^\bot)+\eta&=\{0\}+\eta=\{0\}.
  \end{align*}
Hence
  \begin{align*}
  	&\ell_{x}=\ell'_{x}=\dim(V/\xi)=k,\\
  	&\ell_{z}=\ell'_{z}=0,\\
  	&\ell_{x}+\ell'_{z}=\ell_{z}+\ell'_{x}=k.
  \end{align*}   
\end{example}  

\subsection{A homological formulation}
\label{ssec:QCSS-chain}

As we mentioned in Introduction, there is a mutually beneficial interplay between Quantum Coding Theory, particularly with regard to the CSS codes, and Algebraic Topology. 
The table in \cite[p.~12]{Bravyi:HMP:2014} 
supplies a dictionary between the two theories. 
Making use of some entries of that dictionary, we 
now interpret
the CSS version of CL as given in Proposition~\ref{prop:CL-CSS}. 

Consider a chain complex $\mathcal{C}$,
\begin{equation}
\label{complex-ch}
C_2\overset{\partial_2}{\longrightarrow}C_1\overset{\partial_1}{\longrightarrow}C_0,
\end{equation}
and the corresponding cochain complex
\begin{equation}
\label{complex-coch}
C^2\overset{\delta^2}{\longleftarrow}C^1\overset{\delta^1}{\longleftarrow}C^0.
\end{equation}
Here $C_i$ is an $n_i$-dimensional space over $\FF_2$, $C^i$ is the dual space of $C_i$, 
the boundary and coboundary operators $\partial_i$ 
and $\delta^i$ are dual to each other. By definition, $\partial_1\partial_2=0$; equivalently $\delta^2\delta^1=0$. Suppose that for each $i$ the biorthogonal bases in $C_i$ and $C^i$ are fixed and, accordingly, the isomorphisms between $C_i$ and $C^i$ are fixed. The operators $\delta^i$ and $\partial_i$ can be identified with their matrices.

To a given CSS code $Q$ of length $n$ defined by matrices $H_x$ of size $m_1\times n$ and $H_z$ of size $m_2\times n$
we put in correspondence the complexes \eqref{complex-ch}, \eqref{complex-coch} 
with $n_0=m_1$, $n_1=n$, $n_2=m_2$, 
$\partial_1=(\delta^1)^T=H_x$ and $\partial_2^T=\delta^2=H_z$.
The condition $\partial_1\partial_2=0$
takes the form of Eq.~\eqref{Hxz0}.  The space of degenerate $X$-codewords is $\xi=\img H_x^T=\img \delta^1$, the space of degenerate $Z$-codewords is $\eta=\img H_z^T=\img \partial_2$. The set of $X$-codewords is $\ker H_z=\ker \delta^2=\eta^\bot$, and the set of $Z$-codewords is $\ker H_x=\ker \partial_1=\xi^\bot$. Then the space of logical $X$-operators (up to equivalence) is $\eta^\bot/\xi=\ker \delta^2/\img\delta^1=H^1(\mathcal{C};\ZZ_2)$, i.e.\ the cohomology group of the complex $\mathcal{C}$; the space of logical $Z$-operators is $\xi^\bot/\eta=\ker \partial_1/\img \partial_2=H_1(\mathcal{C};\ZZ_2)$, i.e.\ the cohomology group of the complex $\mathcal{C}$. The dimension $k$ of the code $Q$ is equal to the dimension of the space of its logical $X$- (or $Z$-) operators, 
i.e. $k=\dim H_1(\mathcal{C};\ZZ_2)$. 

For a subspace $\alpha\subset C_1$ we denote by $\hclass{\alpha}=(\alpha\ncap \xi^\bot+\eta)/\eta$ the set of homology classes in $H_1(\mathcal{C};\ZZ_2)$ that can be represented by cycles contained in $\alpha$. Similarly, for a subspace $\beta\subset C^1$ we denote by $\hclass{\beta}=(\beta\ncap \eta^\bot+\xi)/\xi$ the set of cohomology classes in $H^1(\mathcal{C};\ZZ_2)$ that can be represented by cocycles contained in $\beta$. By Proposition \ref{prop:factor-isom}, for any $\alpha\subset C_1$ we have 
$$
   \hclass{\alpha}^\bot = \hclass{\alpha^\bot},
$$
where the symbol ${}^\bot$ in the left-hand side stands to denote the annihilator corresponding to the duality between $H_1(\mathcal{C};\ZZ_2)$ and $H^1(\mathcal{C};\ZZ_2)$.

This means that $\alpha^\bot$ contains representatives of all
logical $X$-operators that commute with all logical $Z$-operators that have representatives in $\alpha$.

For dimensions we have 
$$
 |\hclass{\alpha}|+|\hclass{\alpha^\bot}|=\left|\hclass{\alpha}\right|+|\hclass{\alpha}^\bot|=\dim H_1(\mathcal{C};\ZZ_2)=k,
$$
which is equivalent to the identity stated in Proposition~\ref{prop:CL-CSS}.

\subsection{The case of Subsystem Codes}
\label{ssec:SubC}

Subsystem Codes (SubC) \cite[Sec.~1.3 and 3]{Bravyi&Terhal:2009}, \cite{Haah:2012} are a generalization of Stabilizer Codes.
At the primary level, a new constituent is the {\em gauge group}\ $G$, which is a subgroup of the Pauli group
like the stabilizer group $S$ in the case of QSC, but generally non-commutative.
At the secondary level, it corresponds to a general (non-isotropic) subspace of the symplectic space~$P$.

The stabilizer group $S$ in this context is the center of the group $G$.
Conversely, one can interpret $G$ as the group generated by $S$ and the Pauli operators acting on some number $g$
of  ``gauge qubits''. 
At the secondary level, $S$ corresponds to the radical, $G\cap G^\bot$, of the subspace $G$.
QSC is formally realized as a particular case of SubC when the gauge group is Abelian: $G=S$, that is,  $g=0$.

\smallskip
A dictionary between the primary and secondary levels for SubC is presented in the following
table (notation mostly corresponds to that in \cite{Haah:2012}). 

\bigskip\noindent
\begin{tabular}{p{0.35\textwidth}|l|c}
\hline
Gauge group $G\leq P$ & Subspace $G\subset P$ & $n-k+g$ \\ 
$\mathcal C(G)$ --- centralizer of $G$ in the Pauli group & Subspace $G^\bot$ & $n+k-g$ \\
Stabilizer group $S$ of the code $=$ center of $G$ &
Isotropic subspace $S=G\cap G^\bot$ & $n-(k+g)$\\
Algebra of logical operators $\langle\mathcal{C}(S)\rangle$ & Subspace $S^\bot=G^\bot+G$ & $n+k+g$\\
Bare logical operators 
& $G^\bot/S$ & $2k$\\ 
Dressed logical operators
& $S^\bot/G$ & $2k$\\
Operators supported on $M$ & Subspace $\alpha\subset P$ & \\ 
Operators supported on $M^c$ & Subspace $\alpha^\bot$ & \\ 
Dressed logical operators supported on $M$ & $(S^\bot\cap\alpha+G)/G$ & $g(M)$ \\
Bare logical operators supported on $M^c$  & $(G^\bot\cap\alpha^\bot+S)/S $ & $g_{\text{bare}}(M^c)$  \\
\hline
\end{tabular}

\medskip
Note that the spaces of ``dressed''
and ``bare'' logical operators are canonically isomorphic at the symplectic space level: $G^\bot/S\cong(G^\bot+G)/G\cong  S^\bot/G$, however, their versions with restricted support (in the last two lines) are not.

\medskip
The Cleaning lemma for a subsystem code with gauge group $G$ \cite[Lemma~2, p.~14]{Bravyi&Terhal:2009} asserts the alternative: if $M$ is an arbitrary subset $M$ of qubits, then

{\em
\begin{itemize}
    \item[---] 
    either there exists a nontrivial dressed logical operator supported on $M$;
    \item[---]
     or any logical operator, up to equivalence modulo $S$,
     can be realized as an operator supported on $M^c$.
\end{itemize}
} 

\medskip
Similarly to Proposition~\ref{prob:CL-stabilizer}, there is a formulation
at the secondary level of a more general dimension identity, which implies CL when $g(M)=0$;
this is identical to the first part of Lemma~2 in \cite{Haah:2012}, except that $M$ and $M^c$
are interchanged.

\begin{proposition}
\label{prop:CL-subsystem}
Given are a $2n$-dimensional symplectic space $P$, a subspace $G\subset P$, and a partition $[n]=M\sqcup M^c$.
Using notation as in the 2nd and 3rd columns of the table above, we have the idenity.
$$
g_{\text{bare}}(M^c)+g(M)=2k.
$$
\end{proposition}

\begin{proof}
Put $\xi=S$ and $\eta=G$ in Theorem~\ref{thm:InnProd}(a). (The dimension of the ambient space, denoted by 
$n$ in Th.~\ref{thm:InnProd}, becomes $2n$ in the present application.)
The condition $G\bot S$ is met, and
$|S^\bot|-|G|=2k$. Hence the result.
\end{proof}

\begin{exercise}
The second part of Lemma~2 in \cite{Haah:2012}  states a dimension identity for {\em subsystem CSS codes}. Obtain it as a corollary of results of our Sec.~\ref{ssec:lalg}.
\end{exercise}

\section{Miscellaneous variations}
\label{sec:misc}

We give a few further applications of the theory developed in Sec.~\ref{sec:genrk}. These examples may not be so important for their own sake, but we invite the readers to view them as a ``food for thought''.

\subsection{Fredholm's Alternative
and the Finite-dimensional \\Index Theorem}
\label{ssec:misc-FA}

To get a feel of a proposed novel approach, to decide whether it is deep or shallow, it helps to test it in a well-familiar situation. We offer such an example here. A conclusion is left to the judgement of the reader.

Fredholm's Alternative states that for a linear endomorphism $A:V\to V$ of a finite-dimensional vector space exactly one of the following is true:

\begin{itemize}
\item[(i)] 
 There exists a vector $x\neq 0$ such that $Ax=0$.

\item[(ii)]
For any vector $y\in V$
there is a vector $x$ such that $Ax=y$.
\end{itemize}

A quantitative generalization of this statement is the finite-dimensional Index Theorem. We will derive it as a consequence of our
abstract Theorem~\ref{thm:Agraded-2lat}. To make things a little bit more interesting, we consider vector spaces over a field $\kk$ which is not necessarily commutative. If operators are represented by matrices, the theorem says that the left column rank of a rectangular matrix is equal to its right row rank, cf.\ \cite[Sec.~I.5]{Artin1958}.

Recall that if $V$ is a left (respectively, right)
vector space over $\kk$, then its dual, $V^*$, defined as the space of linear maps $V\to\kk$, is a right (resp., left) vector space over $\kk$.
If $A:V\to W$ is a linear map between left vector spaces, then its transpose is the map 
$A^T: W^*\to V^*$ defined by
$$
 \langle v,A^T w'\rangle=\langle Av,w'\rangle 
$$
for any $v\in V$, $w'\in W^*$.

\begin{theorem}[Finite-dimensional index theorem]
Consider left vector spaces $V$ and $W$ over a field $\kk$, which is not assumed to be commutative.
Suppose that
$m=\dim V$ and $n=\dim W$.
Let $V^*$ and $W^*$ be the respective dual spaces (which are right vector spaces over $\kk$). 
Let $A:V\to W$ be a linear map and $A^T:W^*\to V^*$
its transpose. Then
$$
 \dim\ker A-\dim\ker A^T=m-n.
$$
\end{theorem}

\begin{proof}
There is a natural duality between the left space $V\oplus W$
and the right space $V^*\oplus W^*$ defined by the pairing
$$
 \langle v\oplus w,v'\oplus w'\rangle=
\langle v,v'\rangle+\langle w,w'\rangle.
$$
Consider the modular lattices $L_1=\Gr(V\oplus W)$ (left Grassmanian)
and 
$L_2=\Gr(V^*\oplus W^*)$ (right Grassmanian).
The map $$L_1\ni\zeta\mapsto\zeta^\bot\in L_2$$ 
is an order-reversing bijection from $L_1$
onto $L_2$. 

Consider the subspaces
$$
    \xi=V\oplus 0_W \in L_1, \qquad \eta=V^*\oplus 0_{W^*}\in L_2,
$$
and the graph of $A$, 
$$
\alpha=\{v\oplus Av\mid v\in V\}\in L_1.
$$

Then
$$
    \xi^\bot=0_{V^*}\oplus W^*, \qquad \eta^\bot=0_{V}\oplus W,
$$

$$
 [\eta^\bot/\xi]_1=[\xi^\bot/\eta]_2 =n-m.
$$

Next,
$$
    \xi\cap\alpha=\ker A\oplus 0_W, \qquad \eta^\bot\cap\alpha=0_{V\oplus W},
$$
hence
$$
 \zdim{\eta^\bot}{\xi}{\alpha}_1=-\dim\ker A.
$$

Since
$$
\alpha^\bot=\{(-A^Tw',w')\mid w'\in W^*\}\in L_2,
$$
we obtain similarly
\begin{equation*}
    \xi^\bot\cap\alpha^\bot=0_{V^*}\oplus\ker A^T, \qquad \eta\cap\alpha^\bot=0_{V^*\oplus W^*},
\end{equation*}
so
$$
 \zdim{\xi^\bot}{\eta}{\alpha^\bot}_2=\dim\ker A^T.
$$
Theorem~\ref{thm:Agraded-2lat} yields the required identity.
\end{proof}

\subsection{Some further applications to quantum codes}
\label{ssec:misc-QC}

In this section we consider two propositions that are not versions or generalizations of CL, yet
are also related to QECC and are proved the technique developed in Sec.~\ref{sec:CL}.

The first proposition implies that there are ``universal'' subspaces in $\FF_2^n$
containing many logical operators (more precisely, representatives of the corresponding factorclasses) for {\em any}\ CSS code (\S~\ref{ssec:QCSS}) of the given length and dimension. This fact has no analogs in the classical coding theory.

\begin{proposition}
\label{prop:logops-lagrangian}
Let $Q$ be a quantum CSS code of length $n$ and dimension $k$. If $\alpha\in \FF_2^{n}$ is a subspace such that $\alpha=\alpha^\bot$, then $\alpha$ contains (representatives of) $k$ independent  logical operators of the code $Q$.
\end{proposition}
\begin{proof}
    Let $\xi$ and $\eta$ be the spaces of degenerate $X$-codewords and $Z$-codewords of the code $Q$. 
    Recall that $\ell_{X}=|(\eta^\bot\alpha+\xi)/\xi|$ and
    $\ell_{Z}=|(\xi^\bot\alpha+\eta)/\eta|$ are the respective dimensions of equivalence classes of logical $X$- and $Z$-operators that have representatives in the subspace $\alpha$.
    Since $\alpha=\alpha^\bot$, Proposition \ref{prop:CL-CSS} yields
    $\ell_{X}+\ell_{Z}=k$.
\end{proof}
Note that $\alpha=\alpha^\bot$ is possible only when $n$ is even; then $\dim\alpha=n/2$. The simplest example of a ``universal'' subspace $\alpha=\alpha^\bot$ for even $n$ is $$\alpha=\{(x_1,\dots,x_{n/2},x_1,\dots,x_{n/2})\mid x_i\in\FF_2\}.$$
Proposition \ref{prop:logops-lagrangian} says that any quantum CSS code of dimension $k$ in $\FF_2^n$ has $k$ independent logical operators in $\alpha$. Unlike in Sec.~\ref{sec:CL}, this subspace $\alpha$ is not described in terms of sets of qubits (supports).

\medskip
The next proposition is not new. It is a coordinate-free version of \cite[Lemma 19]{baspin2021:treewidth-dist}. The quoted lemma, in turn, is a ``purified'' version of the original
statement \cite[Eq.~(14)]{Bravyi:UpperDistDim}, which was fomulated in terms of quantum codes of general type and included complicated notions.
In the works \cite{Bravyi:UpperDistDim,baspin2021:treewidth-dist} 
variants of this lemma are 
used in proofs of an estimate of the dimension of a quantum code in terms of its length and minimum distance.

\begin{proposition}
\label{prop:codedim-abc}
Let $V$ be a symplectic space and $\xi\subset V$ be an isotropic subspace ($\xi\subset\xi^\bot$). Suppose that $\alpha,\beta,\gamma\subset V$ are subspaces such that $\beta\bot\gamma$, $\alpha=(\beta+\gamma)^\bot$, and $\xi^\bot\ncap \alpha\subset\xi$, $\xi^\bot\ncap\beta\subset\xi$. Then
$$|\xi^\bot/\xi|\le |\gamma|.$$
\end{proposition}

\begin{proof}
    Since $|(\xi^\bot\alpha+\xi)/\xi|=|\xi/\xi|=0$, 
    Theorem~\ref{thm:Agraded-2} gives
    $$|\xi^\bot/\xi| = |(\xi^\bot\alpha+\xi)/\xi|+|(\xi^\bot\alpha^\bot+\xi)/\xi|=|(\xi^\bot\alpha^\bot+\xi)/\xi|=|\xi^\bot\alpha^\bot+\xi|-|\xi|.$$
    Since $\alpha^\bot= \beta+\gamma$, we have  $\xi^\bot\alpha^\bot=\xi^\bot(\beta+\gamma)\subset \xi^\bot\beta+\gamma$.
    Using the assumption $\xi^\bot\beta\subset \xi$, we get
     $$|\xi^\bot/\xi|\le |\xi^\bot\beta+\xi+\gamma|-|\xi|=|\xi+\gamma|-|\xi|\le|\gamma|,$$
as required.
\end{proof}

In applications to quantum stabilizer code with stabilizer space $S$ one has $\xi=S$ and
$|\xi^\bot/\xi|=2k$, where $k$ is the dimension of the code. 
The subspace $\alpha$, $\beta$ and $\gamma$ are chosen as the spaces of operators supported, respectively, on qubit sets $A$, $B$ and $C$ such that $A\sqcup B\sqcup C$ is the whole set of $n$ qubits, while $A$ and $B$ are correctable sets (corresponding to the conditions $\xi^\bot\ncap \alpha\subset\xi$ and $\xi^\bot\ncap\beta\subset\xi$).  

\subsection{``Cleaning Lemma'' for a product of abelian groups}
\label{ssec:misc-groups}

In  Proposition~\ref{prop:logops-lagrangian} and the paragraph following its proof we have discussed quantum CSS codes where the usual attributes
were present with one exception:
the ``spatial structure'' of the quantum state space, specifically --- the notion of support --- played no role in the description of the subspace $\alpha$.

Here we exhibit a situation of, roughly speaking, an opposite kind. We want the notions of ``support'' and ``cleaning'' to be defined, but we deviate from the mathematical structure underlying the QECC theory.
This example is partly motivated by the considerations in \cite{AKP2004}.

Let $\CC^*$ denote the multiplicative group of nonzero complex numbers and $\QQ^+$ denote the multiplicative group of positive rational numbers.

Consider a finite abelian group $G$. Its dual (the group of characters), $\hat G=\mathrm{Hom}(G,\CC^*)$, is isomorphic to $G$ (non-canonically). If $H$ is a subgroup of $G$, we denote by $|H|$ the cardinality of $H$. 
We denote by $\Lambda(G)$ the lattice of subgroups of $G$ partially ordered by set inclusion.
Clearly, $|\cdot|$ is
a $\QQ^+$-grading on $\Lambda(G)$. 

\begin{definition}
 A {\em nondegenerate bicharacter}\ on $G$ is a map $G\times G\to\CC$, $(g,g')\mapsto \langle g,g'\rangle$
 such that for any fixed $g$, the map $\varphi_g=\langle g,\cdot\rangle: G\mapsto\CC^*$, is a character of $G$
 and the map $g\mapsto \varphi_g$ is an isomorphism from $G$ onto $\hat G$.
 The bicharacter is {\em involutive}\ if $\langle g,g'\rangle=1\,\Leftrightarrow \langle g',g\rangle=1$.
 \end{definition}

\begin{example}
(a) Let $G=(\ZZ/n\ZZ,+)$, $n\geq 2$. Then $\langle a,b\rangle=e^{2\pi i abm/n}$ is an involutive (symmetric) nondegenerate bicharacter on $G$ for any $m$ coprime with $n$.

\smallskip
(b) Let $G$ be as above and $G_2=G\oplus G$. Then $\langle a_1\oplus a_2,b_1\oplus b_2\rangle=
e^{2\pi i (a_1b_2-a_2b_1)/n}$ is an involutive (anti-symmetric) 
nondegenerate bicharacter on $G_2$.  
\end{example}

An involutive nondegenerate bicharacter on $G$ gives rise to a quasi-comple\-mentation on $\Lambda(G)$:
if $H$ is a subgroup, then $H^\dag=\{g\mid \langle H,g\rangle = 1\}$.
It is easy to check that the requirements of Definition~\ref{def:gmodcomp} are met. 

By Theorem~\ref{thm:Agraded-2}, the following is true.

\begin{proposition}
\label{prop:subgr-basic}
If $H$ is a subgroup such that $H\subseteq H^\dag$ and $B$ is any subgroup, then
$$
|(H^\dag\cap B)/(H\cap B)|\cdot |(H^\dag\cap B^\dag)/(H\cap B^\dag)| = |H^\dag/H|.
$$
\end{proposition}

Let us now exhibit a special case (which, on the other hand, is a generalization) of this model that allows one to talk about elements of $G$ supported on a given set.

Suppose that $\{G_j\}_{j=1}^n$ is a family of abelian groups and $\langle \cdot,\cdot\rangle_j$ a family of 
corresponding involutive nondegenerate bicharacters. 
Let $G=G_1\times\dots\times G_n$.  Elements of $G$ can be written as words of the form $\mathbf{g}=(g_1,\dots,g_n)$
with $g_j\in G_j$. Then
$$
 \langle \mathbf{g},\mathbf{g'}\rangle=\prod_{j=1}^n \langle g_j,g'_j\rangle_j
$$
is an involutive nondegenerate bicharacter on $G$.

Let $M$ be some subset of $[n]$. We say that the element $\mathbf{g}$ is supported on $M$ if $g_j=1_j$ for all $j\in M^c$.

Given $M\subset[n]$, put 
\begin{equation}
\label{BM}
B_M=\prod_{j\in M} G_j\,\times \prod_{j\in M^c} \{1_j\}.
\end{equation}
Since the bicharacters $\langle \cdot,\cdot\rangle_j$ are nondegenerate, we have
$
(B_M)^\dag= B_{M^c}.
$

Proposition~\ref{prop:subgr-basic} takes the form, which parallels versions of Cleaning Lemma
discussed earlier.

\begin{proposition}
\label{prop:CL-ab}
Let $G$ be the direct product of finite abelian groups $G_j$, $j=1,\dots,n$ and ${}^\dag$ be the quasi-complementation on the subgroup lattice $\Lambda(G)$ induced by the bicharacter $\langle\cdot,\cdot\rangle$
on $G$, which is the product of nondegenerate involutive bicharacters on the groups $G_j$.
Suppose $H$ is a subgroup of $G$ such that $H\subseteq H^\dag$.
For any $M\subset[n]$ let 
$\ell_M=|B_M H^\dag/B_M H|$, where $B_M$ is the subgroup \eqref{BM}. Let $\ell_{M^c}=|B_{M^c} H^\dag/B_{M^c} H|$.
Then
$$
\ell_M \ell_{M^c}=|H^\dag/H|.
$$
\end{proposition}

The logical alternative between the cases $\ell_M=1$ and $\ell_M>1$ can be stated as follows.

\begin{corollary}
\label{cor:CL-ab-alt}
Let $G$ and $H$ be as in Proposition~\ref{prop:CL-ab} and $M=\{1,\dots,m\}\subset [n]$.
Exactly of the following is true:
\begin{itemize}
    \item[--] either $(H^\dag \setminus H) \cap B_M\ne \varnothing$, 
 that is, some coset $H^\dag/H$ different from $H$ contains an element of the form $(g_1,\dots,g_m,1_{m+1},\dots,1_n)$ (here $\ell_M>1$);
    \item[--] or
   for any coset  $K\in H^\dag /H$ there exists 
   an element of the form $\mathbf{g}=(1_1,\dots,1_m,g_{m+1},\dots,g_n)$ such that $K=\mathbf{g}H$ (here $\ell_M=1$).
\end{itemize}
\end{corollary}






\begin{thebibliography}{10}
\expandafter\ifx\csname url\endcsname\relax
  \def\url#1{\texttt{#1}}\fi
\expandafter\ifx\csname urlprefix\endcsname\relax\def\urlprefix{URL }\fi
\expandafter\ifx\csname href\endcsname\relax
  \def\href#1#2{#2} \def\path#1{#1}\fi

\bibitem{Haah2016}
J.~Haah, Algebraic methods for quantum codes on lattices, Revista Colombiana de
  Matem\'{a}ticas 50~(2) (2016) 299--349.
\newblock \href {http://dx.doi.org/10.15446/recolma.v50n2.62214}
  {\path{doi:10.15446/recolma.v50n2.62214}}.

\bibitem{KitaevShenVyalyi}
A.~Y. Kitaev, A.~H. Shen, M.~N. Vyalyi, Classical and Quantum Computation,
  American Mathematical Society, USA, 2002.
\newblock \href {http://dx.doi.org/10.1090/gsm/047}
  {\path{doi:10.1090/gsm/047}}.

\bibitem{Kribs2005}
D.~Kribs, A quantum computing primer for operator theorists, Linear Algrbra
  Appl. 400 (2005) 147–167.
\newblock \href {http://dx.doi.org/10.1016/j.laa.2004.11.010}
  {\path{doi:10.1016/j.laa.2004.11.010}}.

\bibitem{LaGuardia2020}
G.~G. La~Guardia, Quantum Error Correction: Symmetric, Asymmetric,
  Synchronizable, and Convolutional Codes, Springer International Publishing,
  Cham, 2020.
\newblock \href {http://dx.doi.org/10.1007/978-3-030-48551-1}
  {\path{doi:10.1007/978-3-030-48551-1}}.

\bibitem{Bravyi&Terhal:2009}
S.~Bravyi, B.~Terhal, A no-go theorem for a two-dimensional self-correcting
  quantum memory based on stabilizer codes, New Journal of Physics 11~(4)
  (2009) 043029 (20pp.).
\newblock \href {http://dx.doi.org/10.1088/1367-2630/11/4/043029}
  {\path{doi:10.1088/1367-2630/11/4/043029}}.

\bibitem{Bravyi:UpperDistDim}
S.~Bravyi, D.~Poulin, B.~Terhal, Tradeoffs for reliable quantum information
  storage in 2d systems, Phys. Rev. Lett. 104 (2010) 050503.
\newblock \href {http://dx.doi.org/10.1103/PhysRevLett.104.050503}
  {\path{doi:10.1103/PhysRevLett.104.050503}}.

\bibitem{baspin2021:treewidth-dist}
N.~Baspin, A.~Krishna, Connectivity constrains quantum codes (2021).
\newblock \href {http://arxiv.org/abs/2106.00765} {\path{arXiv:2106.00765}}.

\bibitem{Ioshida2013:FractalQCodes}
B.~Yoshida, Exotic topological order in fractal spin liquids, Phys. Rev. B 88
  (2013) 125122.
\newblock \href {http://dx.doi.org/10.1103/PhysRevB.88.125122}
  {\path{doi:10.1103/PhysRevB.88.125122}}.

\bibitem{Kitaev:toric:1997}
A.~Y. Kitaev, Quantum Communication, Computing, and Measurement, Springer US,
  Boston, MA, 1997, Ch. Quantum Error Correction with Imperfect Gates, pp.
  181--188.
\newblock \href {http://dx.doi.org/10.1007/978-1-4615-5923-8\_19}
  {\path{doi:10.1007/978-1-4615-5923-8\_19}}.

\bibitem{Bravyi:HMP:2014}
S.~Bravyi, M.~B. Hastings, Homological product codes, in: Proceedings of the
  Forty-sixth Annual ACM Symposium on Theory of Computing, STOC '14, ACM, New
  York, NY, USA, 2014, pp. 273--282.
\newblock \href {http://dx.doi.org/10.1145/2591796.2591870}
  {\path{doi:10.1145/2591796.2591870}}.

\bibitem{Hastings:2020:fiber}
M.~B. Hastings, J.~Haah, R.~O'Donnell, Fiber bundle codes: Breaking the
  {$N^{1/2} \operatorname{polylog}(N)$} barrier for quantum {{LDPC}} codes
  (Oct. 2020).
\newblock \href {http://arxiv.org/abs/2009.03921} {\path{arXiv:2009.03921}}.

\bibitem{LiftedProduct}
P.~A. Panteleev, G.~V. Kalachev, Quantum {LDPC} codes with almost linear
  minimum distance (2020).
\newblock \href {http://arxiv.org/abs/2012.04068} {\path{arXiv:2012.04068}}.

\bibitem{BalancedProduct}
N.~P. Breuckmann, J.~N. Eberhardt, Balanced product quantum codes (2020).
\newblock \href {http://arxiv.org/abs/2012.09271} {\path{arXiv:2012.09271}}.

\bibitem{Freedman:qcodes2manifold}
M.~Freedman, M.~Hastings, Building manifolds from quantum codes, Geometric and
  Functional Analysis\href {http://dx.doi.org/10.1007/s00039-021-00567-3}
  {\path{doi:10.1007/s00039-021-00567-3}}.

\bibitem{Stanley1986}
R.~P. Stanley, Enumerative Combinatorics: Volume 1, Wadworth \& Broocks/Cole,
  Springer, Boston, MA, USA, 1986.
\newblock \href {http://dx.doi.org/10.1007/978-1-4615-9763-6}
  {\path{doi:10.1007/978-1-4615-9763-6}}.

\bibitem{MaslovBook}
P.~Piccione, D.~Tausk,
  \href{https://www.ime.usp.br/\textasciitilde{}piccione/Downloads/MaslovBook.pdf}{A
  Student's Guide to Symplectic Spaces, Grassmannians and Maslov Index},
  Publica{\c{c}}{\~o}es matem{\'a}ticas, Instituto de Matem{\'a}tica Pura e
  Aplicada, 2009.
\newline\urlprefix\url{https://www.ime.usp.br/\textasciitilde{}piccione/Downloads/MaslovBook.pdf}

\bibitem{Artin1958}
E.~Artin, Geometric algebra, Interscience Publ., N.-Y.--London, 1958.

\bibitem{Rotman1998}
J.~Rotman, Galois theory, 2nd Edition, Springer-Verlag, New York, 1998.

\bibitem{Kubiak1999}
T.~Kubiak, Mathematics of Fuzzy Sets: Logic, Topology, and Measure Theory,
  Springer US, Boston, MA, 1999, Ch. Separation Axioms: Extension of Mappings
  And Embedding of Spaces, pp. 433--479.
\newblock \href {http://dx.doi.org/10.1007/978-1-4615-5079-2\_8}
  {\path{doi:10.1007/978-1-4615-5079-2\_8}}.

\bibitem{Haah:2012}
J.~Haah, J.~Preskill, Logical-operator tradeoff for local quantum codes, Phys.
  Rev. A 86~(3) (2012) 032308.
\newblock \href {http://dx.doi.org/10.1103/PhysRevA.86.032308}
  {\path{doi:10.1103/PhysRevA.86.032308}}.

\bibitem{AKP2004}
V.~Arvind, P.~P. Kurur, K.~Parthasarathy, Non-stabilizer quantum codes from
  abelian subgroups of the error group, Quantum Inf. Comput. 4 (2004) 411--436.
\newblock \href {http://dx.doi.org/10.26421/QIC4.6-7-2}
  {\path{doi:10.26421/QIC4.6-7-2}}.

\end{thebibliography}

\end{document}